\documentclass[letterpaper]{article} % DO NOT CHANGE THIS
\usepackage{aaai23}  % DO NOT CHANGE THIS
\usepackage{times}  % DO NOT CHANGE THIS
\usepackage{helvet}  % DO NOT CHANGE THIS
\usepackage{courier}  % DO NOT CHANGE THIS
\usepackage[hyphens]{url}  % DO NOT CHANGE THIS
\usepackage{graphicx} % DO NOT CHANGE THIS
\usepackage{natbib}  % DO NOT CHANGE THIS AND DO NOT ADD ANY OPTIONS TO IT
\usepackage{caption} % DO NOT CHANGE THIS AND DO NOT ADD ANY OPTIONS TO IT
\usepackage{algorithm}
\usepackage{algorithmic}
\usepackage{amsfonts}
\usepackage{multirow}
\usepackage{amsthm}
\usepackage{amsmath}
\usepackage{subcaption}
\newtheorem{definition}{Definition}
\newtheorem{theorem}{Theorem}
\newtheorem{lemma}{Lemma}

\newtheorem{corollary}{Corollary}
\newtheorem{mechanism}{Mechanism}

\usepackage{newfloat}
\usepackage{listings}
\DeclareCaptionStyle{ruled}{labelfont=normalfont,labelsep=colon,strut=off} % DO NOT CHANGE THIS
\floatstyle{ruled}
\newfloat{listing}{tb}{lst}{}
\floatname{listing}{Listing}
\title{Collusion-proof And Sybil-proof Reward Mechanisms For Query Incentive Networks}
\author {
    Pingzhong Tang, Youjia Zhang
}
\affiliations {
    Institute for Interdisciplinary Information Sciences\\ 
    Tsinghua University\\
    kenshinping@gmail.com, zhangyou19@mails.tsinghua.edu.cn
}

\usepackage{bibentry}
\begin{document}

\maketitle

\begin{abstract}
     This paper explores reward mechanisms for a query incentive network in which agents seek information from social networks. In a query tree issued by the task owner, each agent is rewarded by the owner for contributing to the solution, for instance, solving the task or inviting others to solve it. The reward mechanism determines the reward for each agent and motivates all agents to propagate and report their information truthfully. In particular, the reward cannot exceed the budget set by the task owner. However, our impossibility results demonstrate that a reward mechanism cannot simultaneously achieve Sybil-proof (agents benefit from manipulating multiple fake identities), collusion-proof (multiple agents pretend as a single agent to improve the reward), and other essential properties. In order to address these issues, we propose two novel reward mechanisms. The first mechanism achieves Sybil-proof and collusion-proof, respectively; the second mechanism sacrifices Sybil-proof to achieve the approximate versions of Sybil-proof and collusion-proof. Additionally, we show experimentally that our second reward mechanism outperforms the existing ones.
\end{abstract}

\section{Introduction}
\label{sec: introduction}

There is an old proverb that states, "Many hands make light work." In other words, the more people involved in a task, the quicker and easier it will be completed. With the growing popularity of social media, many online platforms (e.g., Quora and Stack Overflow) provide the opportunity for people to ask a question online. As opposed to traditional search engines (e.g., Google and Bing), Q\&A platforms allow users to post questions and have them answered by other users rather than a central system that provides several related solutions. A major advantage of online Q\&A platforms is that the task is handled by humans, who are capable of solving more rare questions. Additionally, it is unnecessary for the questioner to simplify the question in order to make it understandable by a machine. 

An important open question in such query models is developing a reward mechanism that ensures the query is propagated successfully and the solution is provided if one exists. Many researchers have discussed and applied such a reward mechanism in various fields, such as peer-to-peer file-sharing systems \cite{golle2001incentives}, blockchains \cite{brunjes2020reward}, and marketing \cite{drucker2012simpler} and task collaboration \cite{P-295}.

In this paper, we mainly focus on reward mechanisms for the answer querying system, in which only one agent is chosen to solve the problem. When implementing such reward mechanisms, it is critical to understand how to motivate agents to spread the information successfully and answer the question honestly. In addition, the total reward should not exceed the budget of the task owner (questioner). Apart from these basic requirements, there are two other significant challenges to be overcome.

The first challenge is Sybil-proof, avoiding agents manipulating multiple false identities to improve the overall reward. For example, an agent who knows the solution can create multiple accounts, and his other accounts can invite him to do so. As the task owner does not know the information of each participant, all these accounts are rewarded. Moreover, Sybil attacks not only increase the monetary costs of the questioner but also delay the time it takes to get a solution. This issue has been discussed in various literature, including \cite{babaioff2012bitcoin, drucker2012simpler, chen2013sybil, lv2016fair, zhang2021sybil}.

The second challenge is collusion-proof, preventing agents from colluding with each other to gain more rewards. For instance, an agent who knows the solution could report it to his parents rather than solve it directly, as such action may increase the overall reward. Additionally, his parents can also take the same action. In consequence, the questioner spends more time solving the problem. A question of this nature is not well studied in the query incentive network \cite{nath2012mechanism, zhang2021sybil}, but has been extensively discussed in other fields, for example, auction design \cite{laffont1997collusion, che2006robustly, marshall2007bidder}.

However, it is impossible for a reward mechanism to simultaneously achieve both Sybil-proof (SP) and collusion-proof (CP) along with a set of other desirable properties (as illustrated in Section \ref{sec: impossibility}). Hence, we not only formally define these properties but also provide approximate versions of SP and CP. In this paper, we propose two novel families of reward mechanisms that are implementable in dominant strategies. The first mechanism is inspired by a geometric mechanism \cite{emek2011mechanisms}. In addition to the basic properties, it also achieves Sybil-proof and collusion-proof, respectively. The second mechanism sacrifices Sybil-proof in order to achieve approximate versions of Sybil-proof and collusion-proof simultaneously. Furthermore, both mechanisms provide a fair allocation/outcome, which is not well-defined in the query incentive network \cite{rahwan2014towards}. Finally, our numerical experiments indicate that our second reward mechanism performs better than the existing ones.

The remainder of the paper is organized as follows. Section \ref{sec: literature review} reviews the relevant literature. Section \ref{sec: preliminaries and model} formally defines the model and properties. In Section \ref{sec: impossibility}, we prove several impossibility results. Thereafter, we present and discuss two families of reward mechanisms in Sections \ref{sec: mechanism 1} \& \ref{sec: mechanism 2}, respectively. In Section \ref{sec: empirical}, we numerically compare our mechanisms with other reward mechanisms. Lastly, we provide concluding remarks and discuss potential future research in Section \ref{sec: conclusion}.

\section{Literature Review}
\label{sec: literature review}

Our research contributes to three streams of literature: query incentive networks, Sybil-proof, and collusion-proof. Below is a brief overview of the research areas closest to our own. 

The seminal paper \cite{kleinberg2005query} was the first to introduce the query incentive model, which treated the query network as a simple branching process. They found that the reward function exhibited threshold behavior based on the branching parameter $b$. Apart from the fixed-payment contracts, \cite{cebrian2012finding} studied the winning team of the Red Balloon Challenge \cite{pickard2011time} and proposed a split contract-based mechanism, which is robust to agents’ selfishness. 

It should be noted that the aforementioned works have not been examined in the context of Sybil attacks. The idea of Sybil attacks was first introduced by \cite{douceur2002sybil}. \cite{douceur2007lottery} formalized a set of desirable properties and proposed the Lottery Tree mechanism, which is Sybil-proof and motivates agents to join in a P2P system. \cite{emek2011mechanisms, drucker2012simpler} studied Sybil-proof mechanisms for multi-level marketing in which agents buy a product and are rewarded for successful referrals. \cite{babaioff2012bitcoin} studied a similar problem and proposed a reward scheme for the Bitcoin system. \cite{lv2015incentive, lv2016fair} introduced a reward mechanism in crowdsourcing or human tasking
systems. \cite{chen2013sybil} proposed a Sybil-proof Mechanism for query incentive networks with the Nash equilibrium implementation, while \cite{zhang2021sybil} considered the dominant strategy.

As we mentioned before, collusion-proof is not well-studied in query incentive networks. \cite{chen2013sybil} left it as future work, and \cite{nath2012mechanism, zhang2021sybil} discussed the impossibility of a reward mechanism to achieve Sybil-proof and collusion-proof simultaneously. Specifically, \cite{nath2012mechanism} proposed a collusion-proof reward mechanism. Many other works such as \cite{laffont1997collusion, che2006robustly, marshall2007bidder} belong to this category. According to related studies on auction design, this research area is promising and deserves further investigation.

\section{Preliminaries and Model}
\label{sec: preliminaries and model}

Consider the problem for a task owner (she) $r$ employing a set of agents to solve the task (e.g., Red-balloon type of challenges and InnoCentive). Each agent (he) can solve the job by himself or ask his friends for assistance. Intuitively, the query process works as follows:
\begin{enumerate}
    \item The task owner $r$ announces the task and the reward scheme and propagates the information to her friends.
    \item Upon receiving the query, her friends decide whether to do the task if they can and whether to continue propagating the query to their friends for assistance.
    \item Once the task is solved, the agents are rewarded according to the rules specified in Step (1).
\end{enumerate}

Following the query process, a query tree rooted at $r$ is built; we denote such a tree as $T_{r}=(V, E)$, where $V$ is the set of all agents, including the task owner $r$, and edge $e(i, j) \in E$ means that agent $i$ informs the task to agent $j$. We also use the standard tree notions of \textit{parent} and \textit{child} in their natural sense; the parents and children of agent $i$ are denoted as $p_{i}$ and $c_{i}$, respectively. Note that an agent $i$ is reachable from the root $r$ if all his ancestors are in the query tree $T_{r}$ and decide to propagate the query. Furthermore, we assume that both owner and agents have zero cost to inform others in our main model for analytical brevity.

In such a process, agents are asked to report two pieces of information, the response to the task $resp_{i} \in \{0, 1\}$ ($1$ for answering, otherwise, $0$) and the set of children $c_{i} \subseteq V \setminus r$. Accordingly, an agent $i$'s action in the mechanism is defined as $\theta_{i}^{\prime}=(resp^{\prime}_{i}, c^{\prime}_{i})$. Notice that agents cannot propagate the task to the non-existing child; therefore, $c^{\prime}_{i} \subseteq c_{i}$.

\begin{definition}
    Given a report file $\theta^{\prime}$ of all agents, let the tree generated from $\theta^{\prime}$ be $T_{r}(\theta^{\prime})=(V^{\prime}, E^{\prime}) \subseteq T_{r}$, where $V^{\prime} \in V$ and $E^{\prime} \subseteq E$.
\end{definition}

Given the above setting, the task owner aims to design a reward mechanism that determines how the task is solved and how the rewards are distributed among the rooted tree $T_{r}(\theta^{\prime})$. (Note that agents cannot invite those who are already in the query. Hence, it is a rooted tree rather than a graph.) Furthermore, the mechanism only rewards agents who have made a positive contribution. Mathematically, the formal definition of such a mechanism is defined as follows.

\begin{definition}
    The Reward Mechanism $M$ on the social network is defined by a task allocation path $f: T_{r} \to T_{r}$ ($T_{r}$ is the structure of the tree rooted at $r$), and a reward function $x=(x_{i})_{i \in f(T_{r}(\theta^{\prime}))}$, where $x_{i}: \Theta \to \mathbb{R_{+}}$ and $\theta_{i} \in \Theta$ is the type profile.
\end{definition}

In this paper, we assume the task owner $r$ chooses only one agent to do the task. Such an assumption has no significant bearing on our results but simplifies analysis and exposition. In addition, when there exist multiple agents who can solve the task, the agent is selected via the shortest path to the root $r$. Since creating fake accounts increases the distance from the root to the solution, to some extent, such a task allocation rule can limit Sybil attacks \cite{chen2013sybil}.

\begin{definition}
    Given the agents profile type $\theta$, for each $resp(i)=1$, we define the shortest path from the root $r$ to agent $i$ as $P_{i}=\{r, a_{1}, a_{2}, ..., i\}$. The task is allocated to the agent with $P=\min _{i} P_{i}$, for all $resp(i)=1$. If there exist multiple agents with $P$, then they are selected randomly.
\end{definition}

\subsection{Properties}
\label{subsec: property}

In this section, we define a set of important properties that a reward mechanism $M$ on the social network should satisfy. All these properties are similar and inspired by the Lottery Tree \cite{douceur2007lottery}, the multi-level marketing \cite{emek2011mechanisms}, and query networks \cite{zhang2021sybil}; some of them are generalized to our query model.

Formally, we provide the definitions for a mechanism to be incentive compatible and individually rational. Since agents are rewarded based on their contributions, agents out of the path $P=f(T_{r}(\theta^{\prime}))$ are not rewarded. We mainly focus our study on the agents in the path $P$ and introduce more strict definitions of IR and IC. 

\begin{definition}
\label{def: po & ic}
    The Reward Mechanism $M$ is 
    \begin{itemize}
        \item Profitable Opportunity (PO), if $x(\theta_{i}) > 0$,
        
        \item Incentive Compatible (IC), if $x(\theta_{i}) \geq x(\theta_{i}^{\prime})$,
    \end{itemize}
    for all agents $i \in P\setminus r$ and $\theta_{i}, \theta_{i}^{\prime} \in \Theta$.
\end{definition}

PO (a.k.a. strongly IR) ensures all agents in the path $P$ are rewarded by the owner $r$ for participating in the mechanism. IC promises all these agents to do the task truthfully and propagate it to all children. Meanwhile, the total reward should be bounded; the task owner never wants to reward more than she has.

\begin{definition}
    The Reward Mechanism $M$ is budget balanced (BB) if there exists a constant $\Pi$ such that 
    $$\sum_{i \in P\setminus r}x(\theta_{i}) \leq \Pi,$$
    and strongly BB if $\sum_{i \in P}x(\theta_{i}) = \Pi.$
\end{definition}

A mechanism is strongly BB if the total reward of all agents $i\in P$ is exactly the budget $\Pi$. 

Next, we consider the reward for each agent. For convenience, hereafter, we define the reward to agent $i$ in the path $P$ with length $n$ as $x(i, n)$, where the agent with depth $n$ in $P$ provides the solution to the task, and $1\leq i \leq n$. 

As discussed in the relevant literature, the reward of an agent $i$ should depend on his child $c_{i}$ and the corresponding depth $i$ in the path $P$ to credit the indirect referrals. For the direct referral, agent $i$ should receive at least a certain fraction of his child's reward $x(i+1, n)$. Violation of these properties leads to a failure in propagating the query under certain conditions. Formally, these properties are defined as follows.

\begin{definition}
    The Reward Mechanism $M$ is $\rho-$split, if $x(i, n) \geq \rho x(i+1, n)$, for any $i \in P\setminus r$ and $0< \rho < 1$.
\end{definition}

As we mentioned before, agents are rewarded for their contribution to completing a task. Hence, an agent may benefit from a false-name attack (or Sybil). Specifically, an agent who has the answer to the question creates a fake account and completes the task through it. Consequently, he receives a greater reward for doing so than for reporting truthfully.

\begin{definition}
\label{def: sp}
    The Reward Mechanism $M$ is $\lambda$-Sybil proof ($\lambda$-SP) ($\lambda \in \mathbb{N}^+$) if
    $$x(i, n)\geq \sum_{k=0}^{\lambda}x(i+k, n+\lambda),$$
    for any $i \in P$ and $\lambda \geq 1$.
    
    If $\lambda$-Sybil proof holds for all positive integers $\lambda$, such a mechanism is \textbf{Sybil-proof}.
\end{definition}

A Sybil-proof (SP) mechanism ensures that any agent in the path $P$ cannot benefit from pretending to be multiple agents. Moreover, a mechanism is $\lambda$-SP if it can prevent agents from gaining more rewards by creating $\lambda$ Sybil attacks (manipulating $\lambda$ more identities).

On the other hand, multiple agents can collude together to pretend as a single agent in order to get more rewards from the owner $r$. \cite{nath2012mechanism, zhang2021sybil} discuss such an observation in detail. The formal definition of collusion proof is given as follows.

\begin{definition}
\label{def: cp}
    The Reward Mechanism $M$ is $(\gamma+1)$-collusion proof ($(\gamma+1)$-CP) ($\gamma \in \mathbb{N}^+$) if
    $$x(i, n)\leq \sum_{k=0}^{\gamma}x(i+k, n+\gamma),$$
    for any $i\in P$ and $\gamma \geq 1$.
    
    If $\gamma$-collusion proof holds for all positive integers $\gamma$, such a mechanism is \textbf{collusion-proof}.
\end{definition}

A collusion-proof (CP) mechanism promises that all agents are worse off from forming a coalition group of any size. Similarly, a mechanism is $\gamma$-CP if agents get more reward from creating a group with a size no less than $\gamma$.

\section{Impossibility Theorem}
\label{sec: impossibility}

So far, we have defined the set of desirable properties that a reward mechanism should satisfy. Our new mechanisms are developed in response to the following impossibility results. One of the impossibility results suggests that a mechanism cannot satisfy profitable opportunity (PO), Sybil-proof (SP), and Collusion-proof (CP) simultaneously.

\begin{theorem}
\label{thm: impossibility 1}
    For $n\geq 3$, there is no Reward Mechanism that can achieve PO, SP, and CP simultaneously. 
\end{theorem}

\begin{proof}
    The proof is similar to those of \cite{nath2012mechanism, zhang2021sybil}. Assume there exists a mechanism that satisfies all three properties. According to Definition \ref{def: sp}, if the mechanism is SP, considering for $m=1$, we have
    \begin{equation}
    \label{eq: sp}
        \begin{aligned}
            x(i, n)&\geq x(i, n+1)+x(i+1, n+1) \\
            &\geq x(i, n+2) + 2x(i+1, n+2) + x(i+2, n+2)
        \end{aligned}
    \end{equation}
    Meanwhile, if the mechanism is also CP, we derive that
    \begin{equation}
    \label{eq: sc}
        \begin{aligned}
            x(i, n)& \leq \sum_{k=0}^{m=2}x(i+k, n+m)\\
            &= x(i, n+2) + x(i+1, n+2) + x(i+2, n+2)
        \end{aligned}
    \end{equation}

    For both Eqs. \ref{eq: sp} \& \ref{eq: sc} to hold, we need $x(i+1, n+2)\leq 0$, which violates to the definition of PO.
\end{proof}

The following theorem explains how $n$ affects the reward of the agent who solves the task for an SP (CP) mechanism.

\begin{theorem}
\label{thm: impossibility 2}
    If the Reward Mechanism is PO and SP (CP), then $x(n, n)$ is non-increasing (non-decreasing) in $n$. 
\end{theorem}

\begin{proof}
    (\textbf{PO, SP and $x(n, n)$ is non-increasing}) We prove this by contradiction. Assume the mechanism is PO, SP, and $x(n, n)$ is increasing in $n$. Since the mechanism is SP (Definition \ref{def: sp}), for any integer $\lambda$, we have
    \begin{equation*}
        \begin{aligned}
            x(n, n) \geq \sum_{k=0}^{\lambda}x(n+k, n+\lambda).
        \end{aligned}
    \end{equation*}
    
    Since the mechanism is also PO (Definition \ref{def: po & ic}), $x(i+k, n+\lambda)>0$, for any $k\in[0, \lambda]$. Therefore, we have $x(n, n)>x(n+\lambda, n+\lambda)$, which contradicts to $x(n, n)$ is increasing in $n$.
    
    (\textbf{PO, CP and $x(n, n)$ is non-decreasing}) Similar to previous part, assume $x(n, n)$ is decreasing in $n$. Since the mechanism is CP (Definition \ref{def: cp}), for any integer $\lambda$, we have $x(n, n) \leq \sum_{k=0}^{\lambda}x(n+k, n+\lambda)$. As the mechanism is PO, $x(i+k, n+\lambda)>0$, for any $k\in[0, \lambda]$. Hence, $x(n, n) <x(n+\lambda, n+\lambda)$, which contradicts to $x(n, n)$ is decreasing in $n$.
\end{proof}

\section{Tree Dependent Geometric Mechanism}
\label{sec: mechanism 1}

Inspired by the work \cite{emek2011mechanisms} in multi-level marketing, in this section, we generalize a geometric reward mechanism into our model, which achieves all basic properties described in Section \ref{subsec: property} and one of SP and CP. This mechanism is called Tree (Topology) Dependent Geometric Mechanism (TDGM).  

\begin{mechanism}[Tree Dependent Geometric Mechanism]
    Given the agents' report file $\theta$ and the corresponding task allocation path $P$, the reward policy of the TDGM is defined as 
    \begin{equation}
    \label{eq: reward function} 
        x(i, n)=\alpha^{n-i}\beta,
    \end{equation}
    for all $i\in P\setminus r$, $0<\alpha<1$, and $0<\beta \leq \frac{1-\alpha}{1-\alpha^{n}} \Pi$.
\end{mechanism}

The understanding of TDGM is intuitive. Each agent is rewarded according to their contribution to the task. Contributions can be divided into two categories: inviting others or solving the issue. For agent $n$ who solves the task is rewarded $\beta$ (we will characterize $\beta$ later). For ancestors of agent $n$, they receive a certain fraction $\alpha$ of the rewards of their children.

\begin{theorem}
\label{thm: mechanism 1 property}
    TDGM satisfies IC, PO, BB, and $\alpha-$split.
\end{theorem}

\begin{proof}
    (\textbf{IC}) We start our proof from IC. Considering the agent $n$ who solves the task, if he does not provide the solution, he would be either in the solution path $P$ or not. If he is still in the path $P$, his reward is $\alpha^{n^{\prime}-i}\beta,$ where $n^{\prime}>n$. Since $\alpha<1$ and $n^{\prime}>n$, $\alpha^{n^{\prime}-i}\beta < \beta$ ($\beta$ is the reward if he solves the task). Moreover, if he is not in the path $P$, he receives nothing from the owner. As a result, agent $n$ is worse off from misreporting. Similarly, for other agents $i \in P \setminus \{r, n\}$. If these agents do not spread the information successfully, they may not be in the path $P$, which leads to no reward for them.
    
    (\textbf{PO \& $\alpha-$split}) Since $0<\alpha<1$, $0<\beta<\frac{1-\alpha}{1-\alpha^{n}}\Pi$ and $0\leq n-i \leq n$, it is obvious that $x(i, n)=\alpha^{n-i}\beta>0$. Moreover, $\frac{x(i, n)}{x(i+1, n)}=\alpha=\rho$.
    
    (\textbf{BB}) Summarizing all the rewards, we derive that
    \begin{equation*}
        \begin{aligned}
            \sum_{i=1}^{n} x(i, n)&=\sum_{i=1}^{n} \alpha^{n-i}\beta\\
            &= \alpha^{n}\beta \cdot \sum_{i=1}^{n} \alpha^{-i}\\
            &= \alpha^{n}\beta \cdot \frac{\alpha^{-1}(1-\alpha^{-n})}{1-\alpha^{-1}}\\
            &= \frac{1-\alpha^{n}}{1-\alpha}\beta
        \end{aligned}
    \end{equation*}
    Since $\beta \leq \frac{1-\alpha}{1-\alpha^{n}} \Pi$, we have $\sum_{i=1}^{n} x(i, n) \leq \Pi$.
    
\end{proof}

Theorem \ref{thm: mechanism 1 property} reveals that TDGM satisfies all basic properties discussed in Section \ref{subsec: property}. 

As a next step, we characterize $\beta$, and analyze the mechanism w.r.t. Sybil-proof and collusion-proof. Based on our discussion of Theorem \ref{thm: impossibility 2}, the reward function of an agent who completes the task should be dependent on 
$n$. Hence, we denote it as a function of $n$ and $\Pi$ such that $\beta=\beta(n, \Pi)$.

\begin{theorem}
\label{thm: mechanism 1 sp & cp}
    TDGM is SP if $\beta(n, \Pi)$ follows 
    \begin{equation}
    \label{eq: condition SP}
        \begin{aligned}
            \beta(n, \Pi)-\beta(n+m, \Pi)\frac{1-\alpha^{m+1}}{1-\alpha} \geq 0,
        \end{aligned}
    \end{equation}
    and it is CP if $\beta(n, \Pi)$ follows 
    \begin{equation}
    \label{eq: condition CP}
        \begin{aligned}
            \beta(n, \Pi)-\beta(n+m, \Pi)\frac{1-\alpha^{m+1}}{1-\alpha} \leq 0,
        \end{aligned}
    \end{equation}
    for any $m \in \mathbb{N}^+$. Furthermore, $0<\beta(n, \Pi)\leq \frac{1-\alpha}{1-\alpha^{n}}\Pi$ in order to satisfy the Budget Balanced condition.
\end{theorem}

\begin{proof}
    To check the SP and CP of TDGM, we expand $\sum_{k=0}^{m}x(i+k, n+m)$ and compare it with $x(i, n)$,
    \begin{equation*}
        \begin{aligned}
            \sum_{k=0}^{m}x(i+k, n+m)&=\sum_{k=0}^{m}\alpha^{n+m-i-k}\beta(n+m, \Pi)\\
            &= \alpha^{n+m-i}\beta(n+m, \Pi)\sum_{k=0}^{m}\alpha^{-k}\\
            &= \alpha^{n+m-i}\beta(n+m, \Pi)\cdot\frac{1-\alpha^{-(m+1)}}{1-\alpha^{-1}}\\
            &= \alpha^{n-i}\frac{1-\alpha^{m+1}}{1-\alpha}\beta(n+m, \Pi).
        \end{aligned}
    \end{equation*}
    If TDGM is SP, by Definition \ref{def: sp}, we have $x(i, n)\geq \sum_{k=0}^{m}x(i+k, n+m)$. Hence, 
    \begin{equation*}
        \begin{aligned}
            \alpha^{n-i}\beta(n, \Pi)-\alpha^{n-i}\frac{1-\alpha^{m+1}}{1-\alpha}\beta(n+m, \Pi)\geq 0\\
            \beta(n, \Pi)-\frac{1-\alpha^{m+1}}{1-\alpha}\beta(n+m, \Pi)\geq 0
        \end{aligned}
    \end{equation*}
    The condition of CP can be derived in a similar way.
\end{proof}

Theorem \ref{thm: mechanism 1 sp & cp} characterizes the reward function of the agent who solves the task in order to satisfy SP and CP, respectively. Furthermore, Theorem \ref{thm: mechanism 1 sp & cp} supplements the result of Theorem \ref{thm: impossibility 2}. 

Interestingly, the Double Geometric Mechanism (DGM) \cite{zhang2021sybil} which is Sybil-proof, is a sub-class of TDGM, and the $\delta$-Geometric Mechanism ($\delta$-GEOM) \cite{nath2012mechanism} which is collusion-proof also belongs to a TDGM with a certain condition. The proof can be found in the full version.

\begin{corollary}
\label{lemma: family}
    The reward mechanisms proposed in \cite{zhang2021sybil} (DGM) and \cite{nath2012mechanism} ($\delta$-GEOM) belong to a family of TDGM.
\end{corollary}

Following then, we check how SP-TDGM (Sybil-proof-TDGM) and CP-TDGM (collusion-proof-TDGM) can achieve approximation versions of collusion-proof and Sybil-proof, respectively.

\begin{lemma}
\label{thm: mechanism 1 approximation sp & cp}
    SP-TDGM is 2-CP, and CP-TDGM can never be SP.
\end{lemma}

\begin{proof}
    If a mechanism can achieve SP and CP simultaneously, then 
    \begin{equation*}
        \begin{aligned}
            x(i, n)=\sum_{k=0}^{m}x(i+k, n+m)\\
            \alpha^{n-i}\beta(n,\Pi)=\sum_{k=0}^{m}\alpha^{n+m-i-k}\beta(n+m, \Pi)\\
            \beta(n, \Pi)=\beta(n+m, \Pi)\frac{1-\alpha^{m+1}}{1-\alpha}\\
            \frac{\beta(n, \Pi)}{\beta(n+m, \Pi)}=\frac{1-\alpha^{m+1}}{1-\alpha}
        \end{aligned}
    \end{equation*}
    
    (\textbf{SP-TDGM is 2-CP}) Since TDGM is SP, by the condition of SP (Eq. \ref{eq: condition SP}), we derive $\frac{\beta(n, \Pi)}{\beta(n+m, \Pi)} \geq \frac{1-\alpha^{m+1}}{1-\alpha}$.
    
    Note that $\frac{1-\alpha^{m+1}}{1-\alpha}$ is increasing in $m$ and $m>0$, thus, $\frac{1-\alpha^{m+1}}{1-\alpha}$ minimized at $m=1$. Intuitively, $x(i, n)=\sum_{k=0}^{m}x(i+k, n+m)$ when $m=1$, and $x(i, n)>\sum_{k=0}^{m}x(i+k, n+m)$ for all integer $m\geq 2$. 
    
    Therefore, SP-TDGM is 2-CP.
    
    (\textbf{CP-TDGM is non-SP}) Since TDGM is CP, by Theorem \ref{thm: impossibility 2}, we have $\beta(n, \Pi) \leq \beta(n+m, \Pi)$. Hence,
    \begin{equation*}
        \begin{aligned}
            \frac{\beta(n, \Pi)}{\beta(n+m, \Pi)} \leq 1 <\frac{1-\alpha^{m+1}}{1-\alpha},
        \end{aligned}
    \end{equation*}
    for any integer $m>1$ and $0<\alpha<1$. Therefore, the equation $\frac{\beta(n, \Pi)}{\beta(n+m, \Pi)}=\frac{1-\alpha^{m+1}}{1-\alpha}$ never holds.
    
    As a result, CP-TDGM is not $\lambda$-SP for any $\lambda$.
\end{proof}

Lemma \ref{thm: mechanism 1 approximation sp & cp} shows that agents benefit from forming a coalition group with a minimum size of two under SP-TDGM, whereas all agents gain more rewards from any Sybil attacks under CP-TDGM.

The \textbf{core} property is widely discussed in cooperative game theory to ensure a stable and fair outcome of the mechanism. Intuitively, a reward allocation has a \textbf{core} property if there are no other coalition groups that can improve agents' rewards. A mechanism is core-selecting if its outcome has a core property. The formal definition is provided in the full version.

\begin{lemma}
\label{thm: mechanism 1 core}
    TDGM is a core-selecting mechanism.
\end{lemma}

\begin{proof}
    The sketch proof is intuitive. Note that the allocation path is selected via the shortest path, and an agent can only join the mechanism via referrals. For agents $i \in P$, if they collude with agents out of path $P$, it increases the length of $P$, and the owner may choose another shortest allocation path $P^{\prime}$. Based on the policy of the mechanism, if they are not chosen in the task allocation path, they are not rewarded by the owner $r$.
    
    Furthermore, as we proved in Theorem \ref{thm: impossibility 2}, for SP-TDGM, the reward for the agent who completes the task is reducing in $n$, which results in a reduction to all his ancestors along path $P$. For CP-TDGM, although the reward increases with $n$, agents can gain a higher reward by Sybil attacks. Thus, rewarded agents in CP-TDGM have no incentive to cooperate with agents out of the path $P$.
    
    Therefore, for agents $i\in P$, no other coalition can improve their utilities under TDGM.
\end{proof}

\section{Generalized Contribution Reward Mechanism}
\label{sec: mechanism 2}

Given the impossibility results in Section \ref{sec: impossibility}, a reward mechanism cannot achieve PO, BB, SP, and CP simultaneously. In Section \ref{sec: mechanism 1}, we introduce a family of geometric mechanisms which satisfy SP and CP, respectively. Indeed, SP-TDGM is 2-CP, and CP-TDGM is not SP.

Nevertheless, the ability to form a collusion group with a size larger than three is limited in a tree network (e.g., agent $i$ not only needs to make a deal with $p_{i}, c_{i}$, but also $p_{p_{i}}, c_{c_{i}}$, and so on.) Therefore, we propose a more practical mechanism that can satisfy Sybil proof and collusion proof to some extent.

\begin{mechanism}[Generalized Contribution Reward Mechanism]
    Given the agents' report file $\theta$ and the corresponding task allocation path $P$, the reward policy of GCRM is defined as
    \begin{equation}
    \label{eq: mechanism 2}
        x(i, n)=\frac{\alpha^{n-i}}{(1+\alpha)^{i}}\beta,
    \end{equation}
    for all $i\in P \setminus r$, $0<\alpha<1$ and $\beta=\Pi$.
\end{mechanism}

In GCRM, each agent is rewarded according to his contribution to the solution, $(\frac{1}{\alpha(1+\alpha)})^{i}$. If $0<\alpha \leq \frac{\sqrt{5}-1}{2}$, the task solver contributes the most and is also rewarded the most, while if $\frac{\sqrt{5}-1}{2} \leq \alpha <1$, the agent who is closest to the task owner contributes the most. In this sense, $\alpha$ may be considered as a parameter to control the contribution between information propagation and task solution. Moreover, the reward of each agent is normalized by a factor $\alpha^{n}$, which depends on the length of the path $P$.

The following theorem shows that GCRM satisfies the basic properties IC, PO, and BB without any restrictions. The proof is similar to that of Theorem \ref{thm: mechanism 1 property} and is provided in the full version.

\begin{theorem}
\label{thm: mechanism 2 property}
    GCRM satisfies IC, PO, and BB.
\end{theorem}

We then analyze the condition of $\alpha$ for GCRM to satisfy the property of $\rho$-split.

\begin{lemma}
    GCRM is $\rho$-split with $\alpha \in (0, \frac{\sqrt{5}-1}{2}]$, where $\rho=\alpha(1+\alpha)\in (0, 1]$.
\end{lemma}

\begin{proof}
    According to the reward function of GCRM (Eq. \ref{eq: mechanism 2}), we have
    \begin{equation*}
        \begin{aligned}
            \frac{x(i, n)}{x(i+1, n)}&=\frac{\frac{\alpha^{n-1}}{(1+\alpha)^{i}}}{\frac{\alpha^{n-i-1}}{(1+\alpha)^{i+1}}}
            &=(1+\alpha)\alpha
            &=\rho
        \end{aligned}
    \end{equation*}
    Since $0<\alpha\leq \frac{\sqrt{5}-1}{2}$, by simple calculation, we derive that $0<\rho \leq 1$, which follows the definition.
\end{proof}
    
So far, we have shown that GCRM is IC, PO, BB, and $\rho$-split under a certain condition. Following then, we study whether the mechanism is Sybil-proof and collusion-proof.

\begin{lemma}
\label{lemma: mechanism 2 sp}
    GCRM is $\lambda^{*}$-Sybil proof, where integer $\lambda^{*}>1$. Furthermore,
    \begin{itemize}
        \item $\lambda^{*}$ increases with $\alpha$,

        \item agents maximize their reward by $\lceil \lambda^{\prime} \rfloor$ Sybil attacks, where $\lambda^{\prime}=\frac{\log(\frac{-\log(1+\alpha)}{\alpha(1+\alpha)(\log(\alpha))})}{\log(\alpha)\log(1+\alpha)}$, ($\lceil \cdot \rfloor$ denotes for the nearest integer function.)
        
        \item the reward more than $\lceil \lambda^{\prime} \rfloor$ Sybil attacks is at most twice as the original one.
    \end{itemize}
\end{lemma}

\begin{proof}
    We start by proving the mechanism is $\lambda^{*}$-SP, where $\lambda^{*}>1$. According to the reward function of GCRM (Eq.\ref{eq: mechanism 2}), we have
    \begin{equation*}
        \begin{aligned}
            \sum_{k=0}^{\lambda}x(i+k, n+\lambda)&=\sum_{k=0}^{\lambda} \frac{\alpha^{n+\lambda -i-k}}{(1+\alpha)^{i+k}}\\
            &= \frac{\alpha^{n+\lambda-i}}{(1+\alpha)^{i}}\sum_{k=0}^{\lambda}\frac{1}{(\alpha(1+\alpha))^{k}}\\
            &= x(i, n)\alpha^{\lambda}\sum_{k=0}^{\lambda}\frac{1}{(\alpha(1+\alpha))^{k}}\\
            &= x(i, n)\frac{1-\alpha^{\lambda+1}(1+\alpha)^{\lambda+1}}{(1+\alpha)^{\lambda}-\alpha(1+\alpha)^{\lambda+1}}
        \end{aligned}
    \end{equation*}
    
    For convenience, hereafter, we denote $f(\alpha, \lambda)=\frac{1-\alpha^{\lambda+1}(1+\alpha)^{\lambda+1}}{(1+\alpha)^{\lambda}-\alpha(1+\alpha)^{\lambda+1}}$.
    
    By substituting $\lambda=1$ into $f(\alpha, \lambda)$, we derive that 
    \begin{equation*}
        \begin{aligned}
            f(\alpha,1)&=\frac{1-\alpha^{2}(1+\alpha)^{2}}{(1+\alpha)^{1}-\alpha(1+\alpha)^{2}}\\
            &= \frac{(1-\alpha(1+\alpha))(1+\alpha(1+\alpha))}{(1+\alpha)(1-\alpha(1+\alpha))}\\
            &= \frac{1}{1+\alpha}+\alpha
        \end{aligned}
    \end{equation*}
    Since $\alpha \in (0, 1)$, $\frac{1}{1+\alpha}+\alpha\in (1, \frac{3}{2})$. As a result, $x(i, n) < x(i, n+1)+x(i+1, n+1)$ is always true for GCRM, which is never a 1-SP mechanism.
    
    Due to space limitations, the rest of the proof is provided in full version.

\end{proof}

\begin{lemma}
\label{lemma: mechanism 2 sp and cp}
    GCRM is $\lambda^{*}$-Sybil-proof and $\lambda^{*}$-collusion proof.
\end{lemma}

\begin{proof}
    As we proved in Lemma \ref{lemma: mechanism 2 sp}, $f(\alpha, \lambda)$ decreases with $\lambda$ if $\lambda>\lambda^{\prime}$. In addition, $\lambda^{*}$ ($\lambda^{*}>\lambda^{\prime}$) is the \textbf{smallest} integer that $f(\alpha, \lambda^{*})\leq 1$. Hence, $\lambda^{*}-1$ is the \textbf{largest} integer that $f(\alpha, \lambda^{*}-1)\geq 1$. Then, we have
    \begin{equation*}
        \begin{aligned}
            \sum_{k=0}^{\lambda^{*}-1}x(i+k, n+\lambda^{*}-1)=f(\alpha, \lambda^{*}-1)x(i, n)\geq x(i, n).
        \end{aligned}
    \end{equation*}
    By Definition \ref{def: cp}, the mechanism is $\lambda^{*}$-CP.
\end{proof}

Intuitively, under GCRM, agents always benefit from at least 1 Sybil attack. Moreover, as $\alpha$ increases, the maximum number of profitable Sybil attacks and the minimum size of the profitable group collusion increase. In contrast to CP-TDGM, under GCRM, there exists an optimal number ($\lambda^{\prime}$) of Sybil attacks to maximize the reward of each agent, and the new reward is at most twice the reward without Sybil attacks.

Despite GCRM can neither prevent Sybil attacks nor collusion, the total reward is upper bounded and decreases with a sufficiently large number of agents. Hence, the total reward never exceeds the budget of the questioner.

\begin{lemma}
\label{lemma: mechanism 2 maximized agents}
    The total reward of GCRM maximized with $\lceil n^{\prime} \rfloor$ agents, where $n^{\prime}=\frac{\log(\frac{-\log(1+\alpha)}{\log(\alpha)})}{\log(\alpha)+\log(1+\alpha)}$.
\end{lemma}

As a result, agents have to trade off between manipulating multiple identities and cooperating with other agents. For example, creating too many Sybil attacks reduces the reward. In the meanwhile, forming a large coalition group is also impractical.

Similarly, GCRM is also a core-selecting mechanism, which ensures the outcome of the mechanism is stable and fair. The proof is the same as in Lemma \ref{thm: mechanism 1 core}.

\begin{lemma}
\label{lemma: mechanism 2 core}
    GCRM is a core-selecting mechanism.
\end{lemma}

\section{Empirical Evaluations}
\label{sec: empirical}

In this section, we empirically evaluate the performance of GCRM and compare it with TDGM. Note that we use DGM \cite{zhang2021sybil} and $\delta$-GEOM \cite{nath2012mechanism} to represent SP-TDGM and CP-TDGM, respectively.

However, $\alpha$ has a different meaning for DGM, $\delta$-GEOM, and GCRM. To keep consistency, we restrict the parameter $\rho$ in order to secure all three mechanisms are $\rho$-split  ($\alpha^{DGM}=\frac{\rho}{1+\rho}$, $\delta=\rho$ and $\alpha^{GCRM}=\frac{\sqrt{1+4\rho}-1}{2}$).

We begin with analyzing the performance of $\delta$-GEOM and GCRM on Sybil-proof, which is graphically shown in Figure \ref{fig: sybil plot}.

\begin{figure}[htp]
    \centering
    \includegraphics[width=0.4\textwidth]{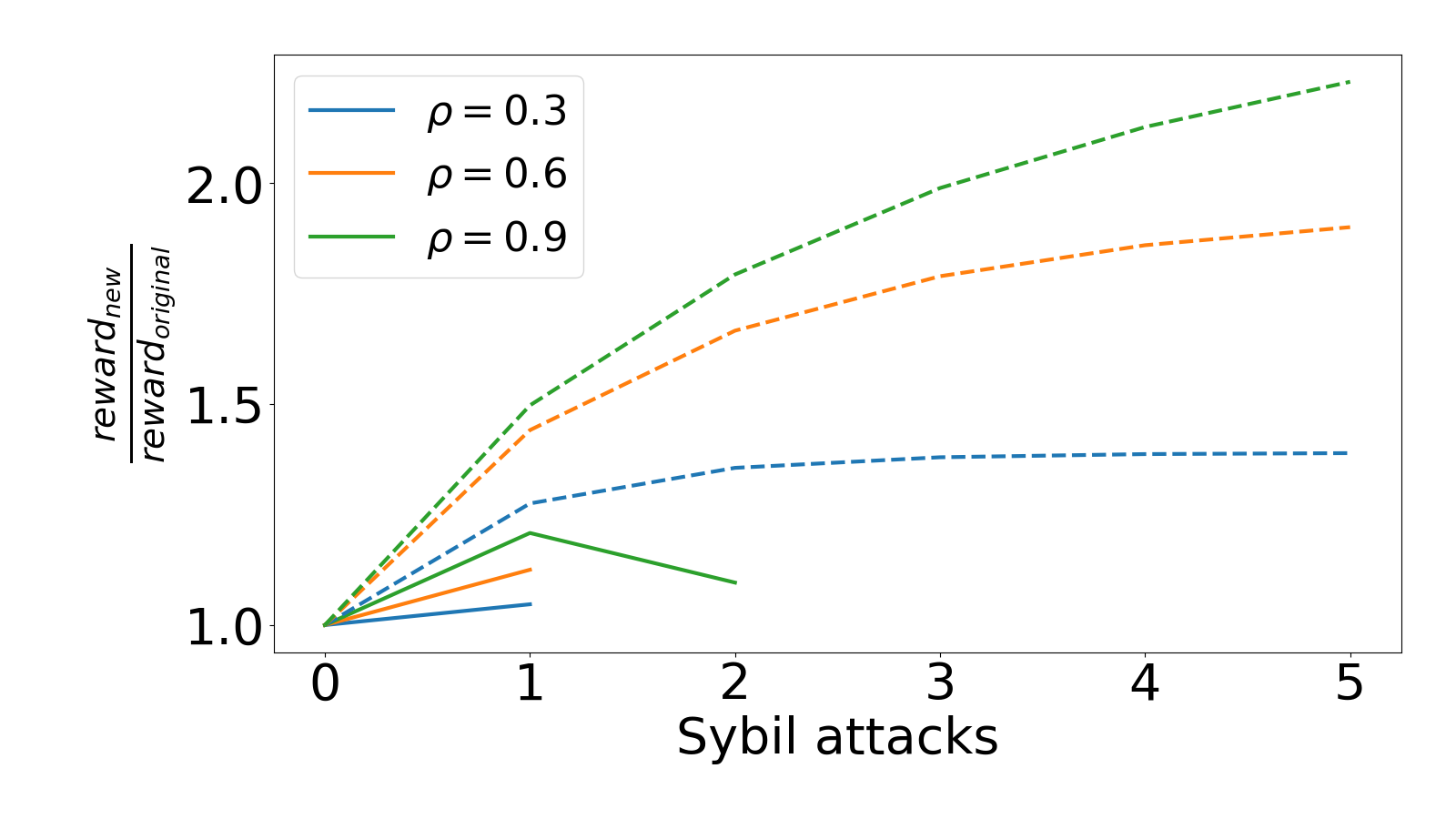}
    \caption{The ratio of new reward to original one after several Sybil attacks. Dashed lines represent $\delta$-GEOM. Solid lines represent GCRM.}
    \label{fig: sybil plot}
\end{figure}

As we can see in Figure \ref{fig: sybil plot}, agents always benefit from manipulating multiple identities under $\delta$-GEOM. In addition, the more Sybil attacks they create, the more reward they gain. However, under GCRM, there exists a $\lambda^{\prime}$ such that agents maximize their reward by $\lceil \lambda^{\prime} \rfloor$ Sybil attacks, and the number of profitable Sybil attacks is upper bounded. Furthermore, agents under GCRM receive less reward from Sybil attacks than under $\delta$-GEOM.

Next, we evaluate the performance of DGM and GCRM on collusion-proof in Figure \ref{fig: collusion plot}. As shown in Figure \ref{fig: collusion plot}, GCRM performs better than DGM on collusion-proof. Regardless of the value of $\rho$, agents under GCRM receive less reward for forming a coalition group than agents under DGM. Furthermore, both Figures \ref{fig: sybil plot} \& \ref{fig: collusion plot} supplement the results of Lemma \ref{lemma: mechanism 2 sp}, $\alpha$ ($\alpha$ is proportional to $\rho$) increases the minimum collusion size requirement to improve the reward.

\begin{figure}[htp]
    \centering
    \includegraphics[width=0.4\textwidth]{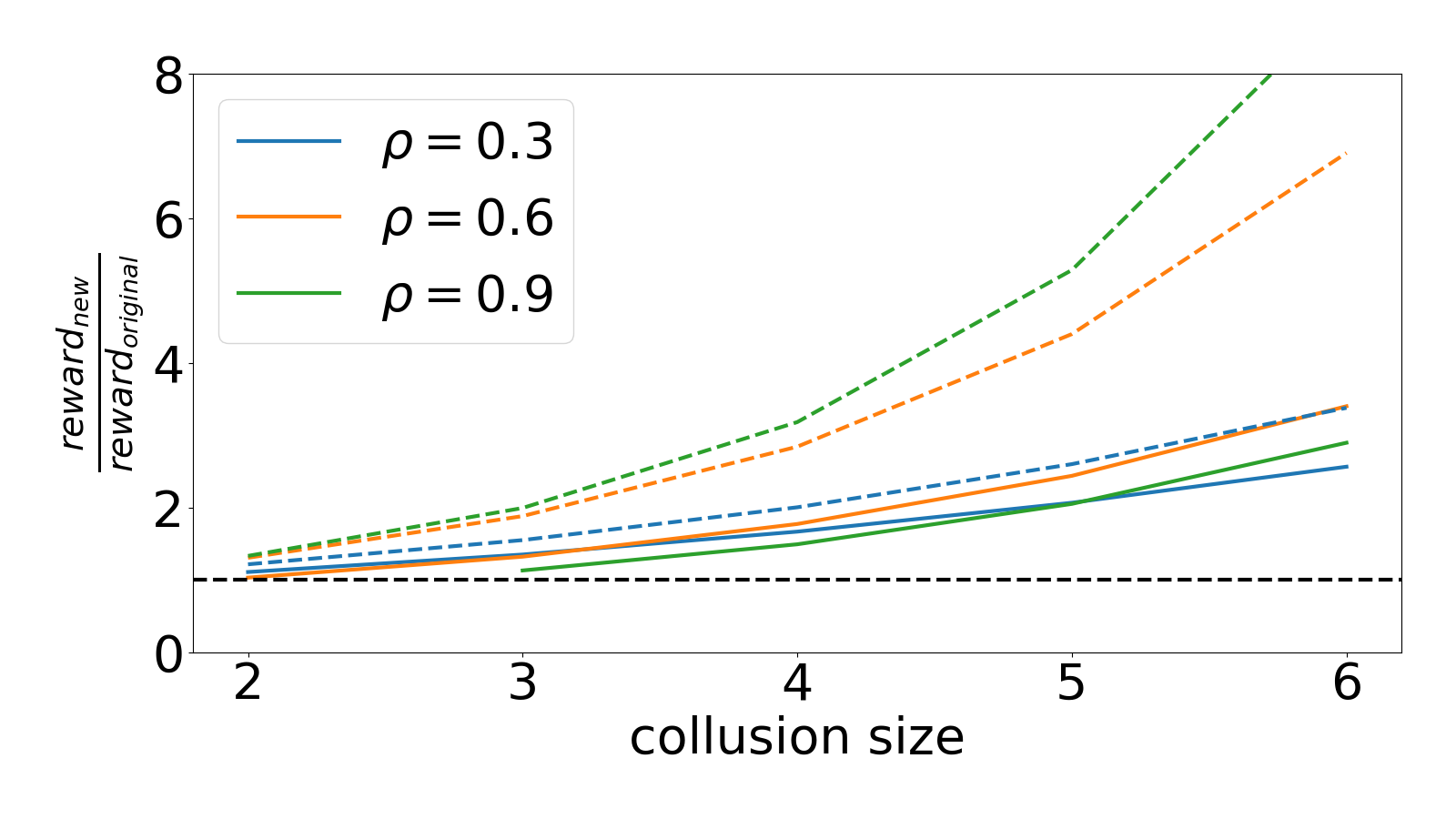}
    \caption{The ratio of new reward to the original one with different collusion size. Dashed lines represent DGM. Solid lines represent GCRM. Black dashed line represents $reward_{new}=reward_{original}$.}
    \label{fig: collusion plot}
\end{figure}

\begin{figure}[htp]
    \centering
    \includegraphics[width=0.4\textwidth]{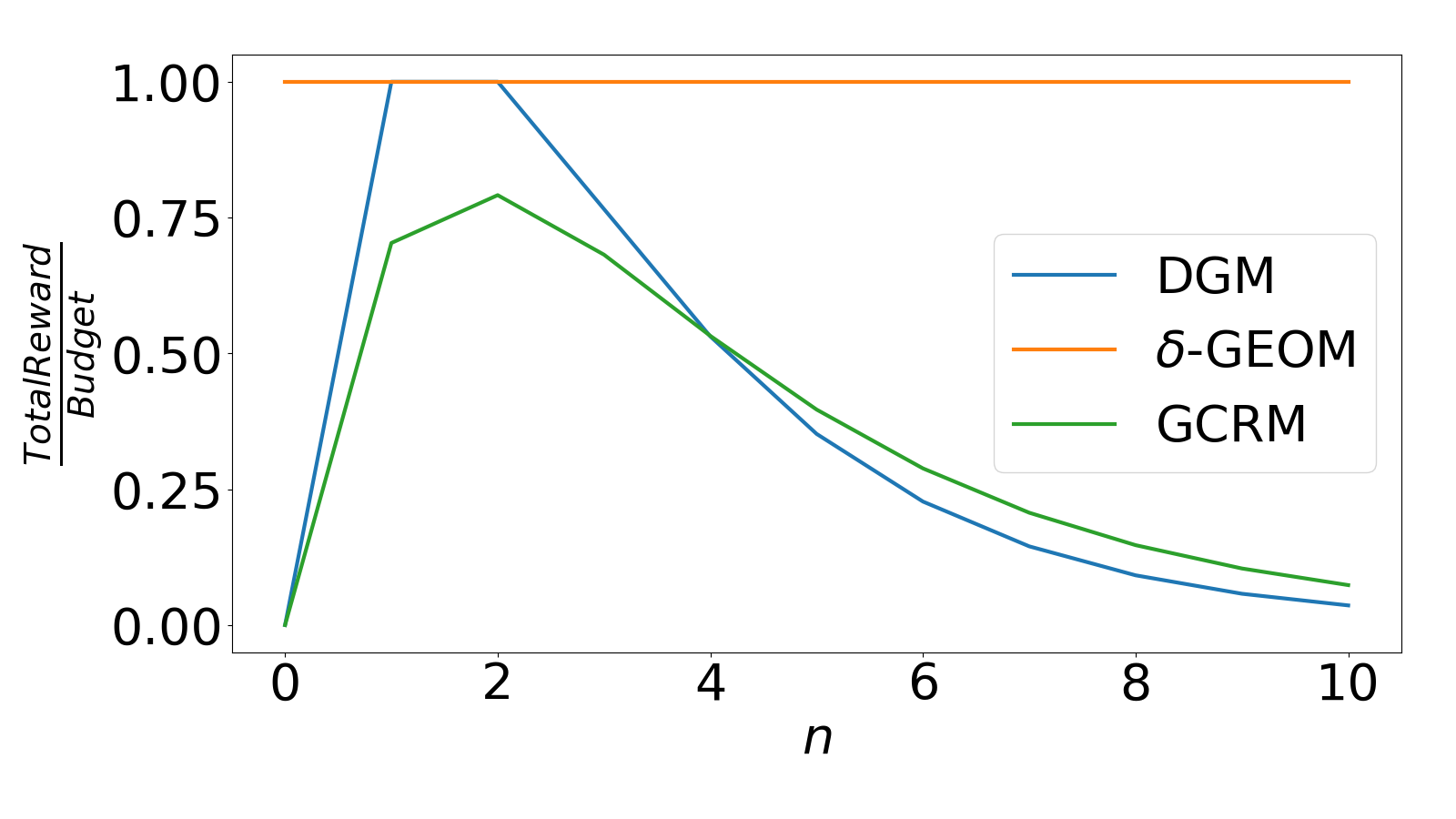}
    \caption{The ratio of total reward to budget, with $\rho=0.6$.}
    \label{fig: reward plot}
\end{figure}

As we proved in Theorems \ref{thm: mechanism 1 property} \& \ref{thm: mechanism 2 property}, all three mechanisms are budget balanced. These patterns are depicted in Figure \ref{fig: reward plot}. 

In conclusion, GCRM allows agents to benefit from specific Sybil attacks, which in turn reduces the efficiency of obtaining the solution. On the other hand, the questioner rewards agents less when they manipulate multiple fake identities. Additionally, GCRM allows agents to benefit from a collusion group of a certain size, which also reduces the efficiency of finding a solution. In practice, however, this size may be too large. In other words, the total reward depends on the efficiency of obtaining the solution; the greater the efficiency, the higher the reward.

Due to space constraints, more details of GCRM are graphically shown in the full version. 

\section{Conclusion}
\label{sec: conclusion}

This paper studies a reward mechanism for a single solution task (or answer query) in a social network. We propose two classes of reward mechanisms to achieve desirable properties, such as PO, IC, BB, SP, and CP. Nevertheless, our impossibility results suggest that a reward mechanism cannot achieve all properties simultaneously. In particular, different questioners may consider different properties to be necessary.

Tree Dependent Geometric Mechanism (TDGM) is a classic geometric mechanism that can achieve Sybil-proof and collusion-proof, respectively. The second mechanism, the Generalized Contribution Reward Mechanism (GCRM), is a more flexible mechanism that sacrifices SP to achieve the approximate versions of SP and CP simultaneously. Specifically, GCRM is $\lambda$-SP and $\lambda$-CP. In other words, the optimal number of profitable Sybil attacks is limited and known to the questioner. The questioner adjusts the parameter $\alpha$ to trade-off between Sybil-proof and collusion-proof in order to improve the efficiency and the cost of obtaining the solution. 

Despite the fact that our research provides some interesting insights into reward mechanisms in a single task allocation, numerous aspects remain to be explored. It is an interesting problem to consider referral costs, such that different agents have different costs to invite others. However, as mentioned in \cite{rochet1998ironing}, mechanism design problems with multi-dimensional type distributions are challenging.

\bibliography{aaai23}

\appendix

\clearpage

\section{Supplementary Material}
Figure \ref{fig: GCRM Sybil} empirically evaluates how $\lambda^{*}$ behaves and individual reward improves after Sybil attacks with different $\alpha$ under GCRM.

\begin{figure}[htp!]
    \centering
    \includegraphics[width=0.45\textwidth]{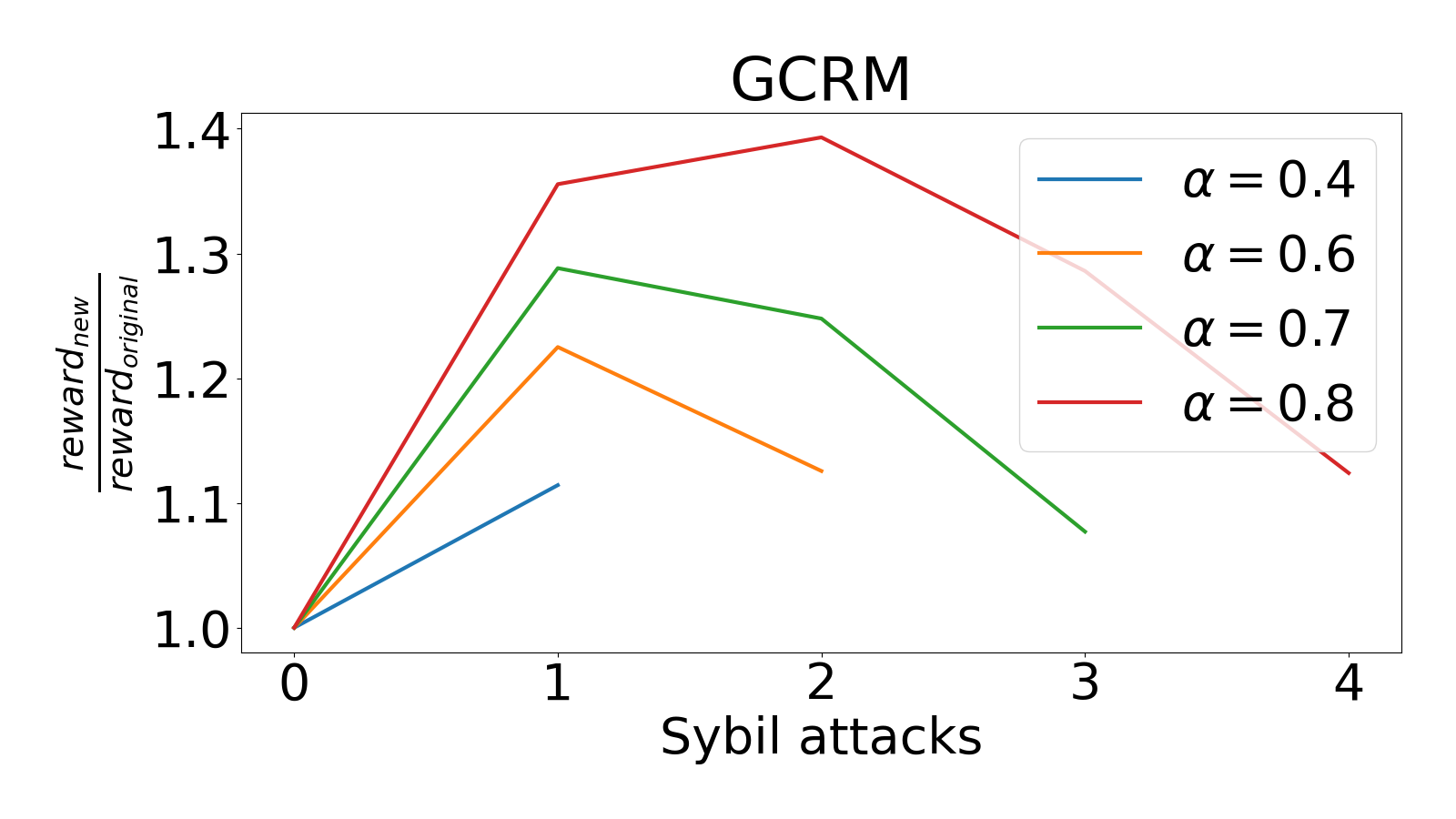}
    \caption{The ratio of new reward to original one after several Sybil attacks. The line end at $\lambda^{*}$. Note that it is unprofitable for $(\lambda^{*}+1)$ Sybil attacks.}
    \label{fig: GCRM Sybil}
\end{figure}

Figure \ref{fig: GCRM Sybil} supplements the result of Lemma \ref{lemma: mechanism 2 sp}. There exists a $\lambda^{\prime}$ such that agents maximize their reward by $\lceil \lambda^{\prime} \rfloor$ Sybil attacks, and both $\lambda^{*}$ and $\lceil \lambda^{\prime} \rfloor$ increase with $\alpha$.

Figure \ref{fig: GCRM collusion} empirically evaluates how individual reward improves with different collusion sizes and $\alpha$ under GCRM.

\begin{figure}[htp]
    \centering
    \includegraphics[width=0.45\textwidth]{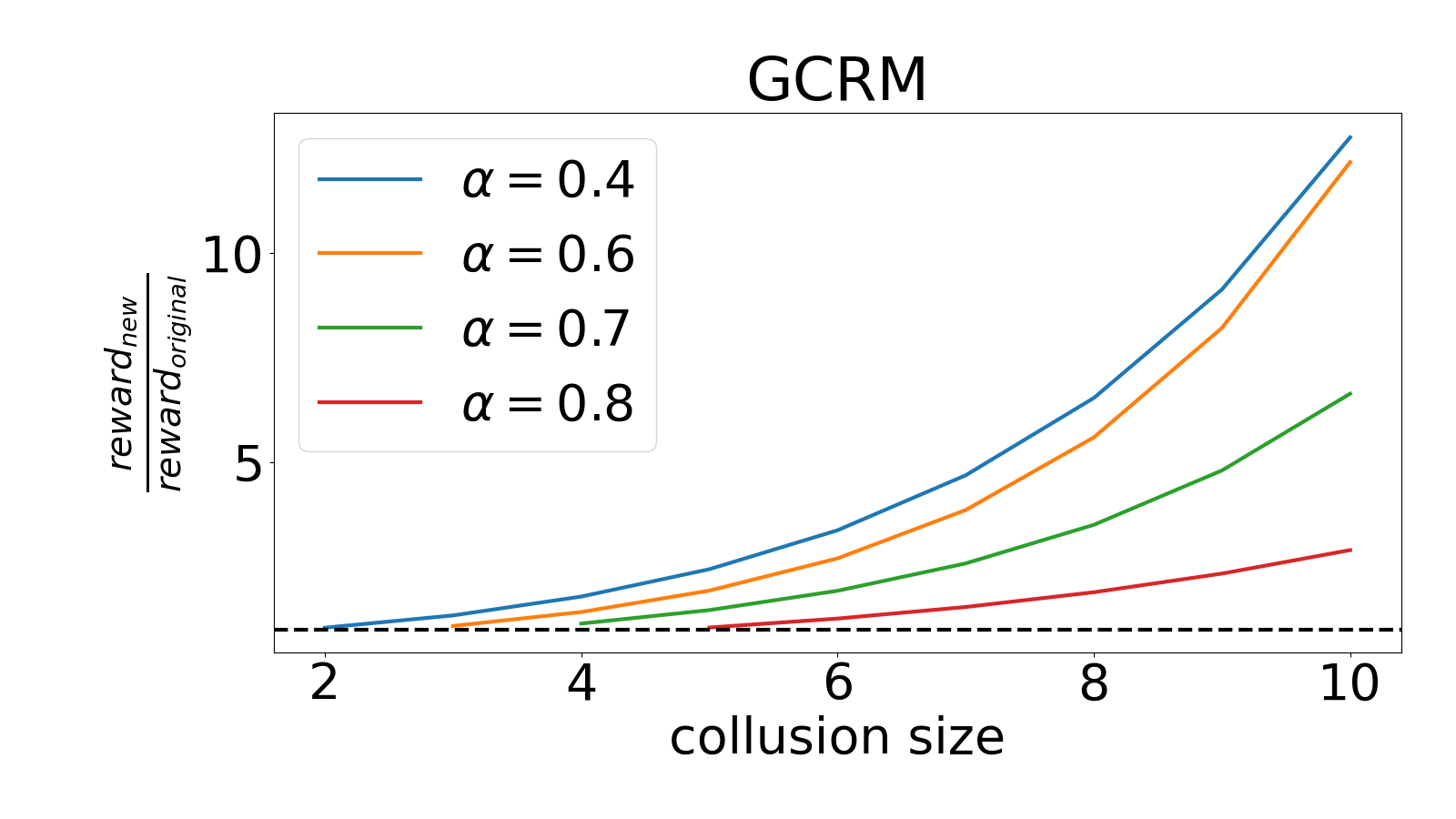}
    \caption{The ratio of new reward to the original one with different collusion size. Black dashed line represents $reward_{new}=reward_{original}$.}
    \label{fig: GCRM collusion}
\end{figure}

Figure \ref{fig: GCRM collusion} supplements the result of Lemma \ref{lemma: mechanism 2 sp}. Moreover, as $\alpha$ increases, forming a coalition group is less profitable.

\section{Proof of Corollary \ref{lemma: family}}
\begin{proof}
    (\textbf{DGM}) The reward function of DGM is defined as $x^{DGM}(i, n)=\alpha_{DGM}^{n-i}(1-\alpha_{DGM})^{i-1}\Pi$. Note the $\alpha_{DGM}$ is different from $\alpha$ in TDGM. Under DGM, it is $\alpha$-split if
    \begin{equation*}
        \begin{aligned}
            \frac{x^{DGM}(i, n)}{x^{DGM}(i+1, n)}=\frac{\alpha_{DGM}}{1-\alpha_{DGM}}=\alpha.
        \end{aligned}
    \end{equation*}
    By substituting $\alpha=\frac{\alpha_{DGM}}{1-\alpha_{DGM}}$ into $x^{DGM}(i, n)$, we derive
    \begin{equation*}
        \begin{aligned}
            x^{DGM}(i, n)&=(\frac{\alpha}{1+\alpha})^{n-i}(1-\frac{\alpha}{1+\alpha})^{i-1}\Pi\\
            &=(\frac{\alpha}{1+\alpha})^{n-i}(\frac{1}{1+\alpha})^{i-1}\Pi\\
            &=\alpha^{n-i}(\frac{1}{1+\alpha})^{n-1}\Pi
        \end{aligned}
    \end{equation*}
    Hence, $\beta(n, \Pi)=(\frac{1}{1+\alpha})^{n-1}\Pi$. By some calculation, we can simply derive $\beta(n, \Pi)-\frac{1-\alpha^{m+1}}{1-\alpha}\beta(n+m, \Pi)\geq 0$, which follows the condition of SP.
    
    (\textbf{$\delta$-GEOM}) The reward function of $\delta$-GEOM is defined as $x^{GEOM}(i, n)=\delta^{n-i}\frac{1-\delta}{1-\delta^{n}}\Pi$. Note that the $\delta$ has the same meaning as $\alpha$ in TDGM. As a result, $\beta(n, \Pi)=\frac{1-\alpha}{1-\alpha^{n}}\Pi$, which is the upper bound of $\beta$. Moreover, $\beta(n, \Pi)-\frac{1-\alpha^{m+1}}{1-\alpha}\beta(n+m, \Pi)\leq 0$, which follows the condition of CP.
\end{proof}

\section{Formal Definition of \textbf{Core} (Lemma \ref{thm: mechanism 1 core} \& \ref{lemma: mechanism 2 core})}
The work is inspired by \cite{rahwan2014towards}.

The characteristic function (coalitional value function) is defined as $W(S)=W(S, resp, c)$ for a coalition $S \subseteq V$ (note: $V$ is the set of all agents and $resp, c$ is the response of each agent). In other words, $W(S)$ maps a coalition group to a payoff value.

An outcome of a game is a pair $(S, \textbf{r})$, where $\textbf{r}=(r_{1}, r_{2},..,r_{n})$ is a payoff vector, $r_{i}$ is the payoff of agent $i\in S$. Moreover, $\sum_{i\in S}r_{i}=W(S)$.

For a given payoff profile $\textbf{r}$, if $\sum_{i\in S}r_{i}<W(S)$, then $\textbf{r}$ is \textbf{blocked} by $S$.

The outcome is \textbf{stable} if no group of agents can receive a payoff greater than what was allocated to them in that outcome, e.g., $\sum_{i\in S}r_{i}\geq W(S)$ $\forall S \subseteq V$. The set of all stable outcomes in a game
is called the \textbf{core} of that game.

\section{Proof of Theorem \ref{thm: mechanism 2 property}}
\begin{proof}
    (\textbf{IC}) For the agent $n$ who solves the task, if he does not provide the solution, he would be either in the solution path $P$ or not. If he is still in the path $P$, his reward becomes from $\frac{1}{(1+\alpha)^{n}}$ to $\frac{\alpha^{n^{\prime}-n}}{(1+\alpha)^{n^{\prime}}}$, where $n^{\prime}>n$. 
    \begin{equation*}
        \begin{aligned}
            \frac{1}{(1+\alpha)^{n}}-\frac{\alpha^{n^{\prime}-n}}{(1+\alpha)^{n^{\prime}}}&=\frac{1}{(1+\alpha)^{n}}(1-\frac{\alpha^{n^{\prime}-n}}{(1+\alpha)^{n^{\prime}-n}})\\
            &>0
        \end{aligned}
    \end{equation*}
    Furthermore, if he is not in the path $P$, he receives nothing from the questioner. As a result, he is worse off from misreporting. Similarly, for other agents $i\in P\setminus\{r, n \}$.
    
    (\textbf{PO}) It is clear that $x(i, n)=\frac{\alpha^{n-i}}{(1+\alpha)^{i}}>0$.
    
    (\textbf{BB}) Summarizing all the rewards, we have
    \begin{equation*}
        \begin{aligned}
            \sum_{i=1}^{n}x(i, n)&=\sum_{i=1}^{n}\frac{\alpha^{n-i}}{(1+\alpha)^{i}}\Pi\\
            &=\alpha^{n}\Pi\sum_{i=1}^{n}(\frac{1}{\alpha(1+\alpha)})^{i}\\
            &=\alpha^{n}\frac{1}{\alpha(1+\alpha)}\cdot\frac{1-(\frac{1}{\alpha(1+\alpha)})^{n}}{1-\frac{1}{\alpha(1+\alpha)}}\Pi\\
            &=\frac{1-\alpha^{n}(1+\alpha)^{n}}{(1+\alpha)^{n}-\alpha(1+\alpha)^{n+1}}\Pi.
        \end{aligned}
    \end{equation*}
    Now, we need to prove $\frac{1-\alpha^{n}(1+\alpha)^{n}}{(1+\alpha)^{n}-\alpha(1+\alpha)^{n+1}}\leq 1$.
    
    If $0<\alpha\leq \frac{\sqrt{5}-1}{2}$, 
    \begin{equation*}
        \begin{aligned}
            1-\alpha^{n}(1+\alpha)^{n}\leq (1+\alpha)^{n}(1-\alpha(1+\alpha))\\
            1\leq (1+\alpha)^{n}(\alpha^{n}+1-\alpha(1+\alpha)),
        \end{aligned}
    \end{equation*}
    which always holds.
    
    If $\frac{\sqrt{5}-1}{2}<\alpha <1 $,
    \begin{equation*}
        \begin{aligned}
            1-\alpha^{n}(1+\alpha)^{n}\geq (1+\alpha)^{n}(1-\alpha(1+\alpha))\\
            1\geq (1+\alpha)^{n}(\alpha^{n}+1-\alpha(1+\alpha)),
        \end{aligned}
    \end{equation*}
    it is always true.
    Therefore, $\frac{1-\alpha^{n}(1+\alpha)^{n}}{(1+\alpha)^{n}-\alpha(1+\alpha)^{n+1}}\leq 1$ holds for any conditions, and $\sum_{i=1}^{n}x(i, n)< \Pi$.
\end{proof}

\section{Proof of Lemma \ref{lemma: mechanism 2 sp and cp}}
\begin{proof}
        1. We start by proving the mechanism is $\lambda^{*}$-SP, where $\lambda^{*}>1$. According to reward function of GCRM (Eq.\ref{eq: mechanism 2}), we have
    \begin{equation*}
        \begin{aligned}
            \sum_{k=0}^{\lambda}x(i+k, n+\lambda)&=\sum_{k=0}^{\lambda} \frac{\alpha^{n+\lambda -i-k}}{(1+\alpha)^{i+k}}\\
            &= \frac{\alpha^{n+\lambda-i}}{(1+\alpha)^{i}}\sum_{k=0}^{\lambda}\frac{1}{(\alpha(1+\alpha))^{k}}\\
            &= x(i, n)\alpha^{\lambda}\sum_{k=0}^{\lambda}\frac{1}{(\alpha(1+\alpha))^{k}}\\
            &= x(i, n)\frac{1-\alpha^{\lambda+1}(1+\alpha)^{\lambda+1}}{(1+\alpha)^{\lambda}-\alpha(1+\alpha)^{\lambda+1}}
        \end{aligned}
    \end{equation*}
    
    For convenience, hereafter, we denote $f(\alpha, \lambda)=\frac{1-\alpha^{\lambda+1}(1+\alpha)^{\lambda+1}}{(1+\alpha)^{\lambda}-\alpha(1+\alpha)^{\lambda+1}}$.
    
    By substituting $\lambda=1$ into $f(\alpha, \lambda)$, we derive that 
    \begin{equation*}
        \begin{aligned}
            f(\alpha,1)&=\frac{1-\alpha^{2}(1+\alpha)^{2}}{(1+\alpha)^{1}-\alpha(1+\alpha)^{2}}\\
            &= \frac{(1-\alpha(1+\alpha))(1+\alpha(1+\alpha))}{(1+\alpha)(1-\alpha(1+\alpha))}\\
            &= \frac{1}{1+\alpha}+\alpha
        \end{aligned}
    \end{equation*}
    Since $\alpha \in (0, 1)$, $\frac{1}{1+\alpha}+\alpha\in (1, \frac{3}{2})$. As a result, $x(i, n) < x(i, n+1)+x(i+1, n+1)$ is always true for the GCRM, which is never 1-SP mechanism.
    
    2. $\lambda^{*}$ can be derived from $f(\alpha, \lambda)=1$.
    \begin{equation*}
        \begin{aligned}
            f(\alpha, \lambda^{*})=1\\
            \frac{1-\alpha^{\lambda^{*}+1}(1+\alpha)^{\lambda^{*}+1}}{(1+\alpha)^{\lambda^{*}}-\alpha(1+\alpha)^{\lambda^{*}+1}}=1\\
            \frac{1-\alpha^{\lambda^{*}+1}(1+\alpha)^{\lambda^{*}+1}}{(1+\alpha)^{\lambda^{*}}}=1-\alpha(1+\alpha).
        \end{aligned}
    \end{equation*}
     Therefore, as $\alpha$ increases, $\lambda^{*}$ also increases.
    
    3. Next, we prove that $\lambda^{\prime}$ is the maximum point of $f(\alpha, \lambda)$ for all $\alpha$.
    Differentiate $f(\alpha, \lambda)$ w.r.t. $\lambda$, we derive that
    \begin{equation*}
        \begin{aligned}
            \frac{\partial f(\alpha, \lambda)}{\partial \lambda}=\frac{(\alpha + 1)^{- n} (\alpha^{n + 1} (\alpha + 1)^{n} \log{(\alpha )} }{\alpha^{2} + \alpha - 1}\\
            +\frac{\alpha^{n + 2} (\alpha + 1)^{n} \log{(\alpha )} + \log{(\alpha + 1 )})}{\alpha^{2} + \alpha - 1}=0\\
            \implies \lambda^{\prime}=\frac{\log(\frac{-\log(1+\alpha)}{\alpha(1+\alpha)(\log(\alpha))})}{\log(\alpha)\log(1+\alpha)}
        \end{aligned}
    \end{equation*}
    
    Taking the second derivative of $f(\alpha, \lambda)$ w.r.t. $\lambda$, we derive that
    
    \begin{equation*}
        \begin{aligned}
            \frac{\partial^{2}f(\alpha, \lambda)}{\partial \lambda^{2}}=\frac{(\alpha + 1)^{- n} (\alpha^{n + 1} (\alpha + 1)^{n} \log{(\alpha )}^{2}}{\alpha^{2} + \alpha - 1}\\
            +\frac{ \alpha^{n + 2} (\alpha + 1)^{n} \log{(\alpha )}^{2} - \log{(\alpha + 1 )}^{2})}{\alpha^{2} + \alpha - 1}
        \end{aligned}
    \end{equation*}
    
    By plugging $\lambda^{\prime}$ into $\frac{\partial^{2}f(\alpha, \lambda)}{\partial \lambda^{2}}$, for any $\alpha = (0, \frac{\sqrt{5}-1}{2}) \cup (\frac{\sqrt{5}-1}{2}, 1)$, we have $\frac{\partial^{2}f(\alpha, \lambda)}{\partial \lambda^{2}}<0$.
    
    Therefore, agents maximize with $\lceil \lambda^{\prime} \rfloor$ Sybil attacks.

    4. Finally, we prove that the reward after $\lambda^{*}$ Sybil attacks can be at most twice as the original reward. Mathematically, it is equivalent to 
    \begin{equation*}
        \begin{aligned}
            \sum_{k=0}^{\lambda^{\prime}}x(i+k, n+\lambda^{\prime}) < 2x(i, n)\\
            \implies f(\alpha, \lambda^{\prime}) \to 2
        \end{aligned}
    \end{equation*}
    Since $f(\alpha, \lambda^{\prime})$ increases with $\alpha$ for $\lambda \in (0, \lambda^{\prime}]$, as $\alpha \to 1$, $\lambda^{\prime} \to \infty$ (L'Hôpital's rule), and we have
    \begin{equation*}
        \begin{aligned}
            f(\alpha, \lambda^{\prime}) &\to \frac{1-2^{\lambda^{\prime}+1}}{2^{\lambda^{\prime}}-2^{\lambda^{\prime}+1}}\\
            & = \frac{(1-2)\sum_{j=0}^{\lambda^{\prime}}2^{j}}{(1-2)2^{\lambda^{\prime}}}\\
            & = \sum_{j=0}^{\lambda^{\prime}} \frac{1}{2^{j}}\\
            &\to 2
        \end{aligned}
    \end{equation*}
\end{proof}

\section{Proof of Lemma \ref{lemma: mechanism 2 maximized agents}}
\begin{proof}
    As a part of the proof of Theorem 5, the total reward of GCRM is
    \begin{equation*}
        \sum_{i=1}^{n}x(i, n)=\frac{1-\alpha^{n}(1+\alpha)^{n}}{(1+\alpha)^{n}-\alpha(1+\alpha)^{n+1}}\Pi.
    \end{equation*}
    By differentiating it w.r.t $n$, we have
    \begin{equation*}
        \begin{aligned}
            \frac{\partial \sum_{i=1}^{n}x(i, n)}{\partial n}=0\\
            \implies n^{\prime}=\frac{\log(\frac{-\log(1+a)}{\log(a)})}{\log(a)+\log(1+a)}
        \end{aligned}
    \end{equation*}
    By substituting $n^{\prime}$ into $\frac{\partial^{2}\sum_{i=1}^{n}x(i, n)}{\partial n^{2}}$, for any $0<\alpha<1$, we have $\frac{\partial^{2}\sum_{i=1}^{n}x(i, n)}{\partial n^{2}}<0$.
    
    Therefore, $n^{\prime}$ is the maximum point of the total reward function of GCRM.
\end{proof}

\end{document}